\documentclass[11pt]{article}

\usepackage[margin=1in]{geometry}
\usepackage[dvipsnames]{xcolor}
\usepackage[pdfusetitle,colorlinks=true,linkcolor=blue,citecolor=BrickRed,bookmarksopen,pdfencoding=auto,psdextra]{hyperref}
\usepackage{mathptmx, amsmath, amsthm, amssymb, mathtools, bm, bbm, xspace, multirow}
\usepackage{algorithm, algpseudocodex, url}

\usepackage{cleveref}

\newtheorem{theorem}{Theorem}[section]
\newtheorem{lemma}[theorem]{Lemma}

\theoremstyle{definition}

\theoremstyle{remark}

\newcommand{\expec}[1]{ \textup{E}\left [ #1 \right ]}
\newcommand{\rounddown}[1]{\left \lfloor  #1 \right \rfloor}
\newcommand{\roundup}[1]{\left \lceil  #1 \right \rceil}

\DeclareMathOperator{\polylog}{polylog}
\DeclareMathOperator{\poly}{poly}
\DeclareMathOperator{\argmax}{{argmax}}

\DeclareMathOperator{\Bern}{Bern}
\DeclareMathOperator{\Bin}{Bin}

\newcommand{\tO}{\tilde{O}}

\newcommand{\id}{\textsc{id}}

\newcommand{\prob}[1]{\Pr\left[{#1}\right]}
\newcommand{\EE}{\mathbb{E}}
\newcommand{\FF}{\mathbb{F}}
\newcommand{\RR}{\mathbb{R}}

\newcommand{\cO}{{\mathcal O}}
\newcommand{\etal}{{et al.~}}

\newcommand{\hatm}{\hat{m}}

\newcommand{\cC}{\mathcal{C}}
\newcommand{\cF}{\mathcal{F}}
\newcommand{\cG}{\mathcal{G}}

\newcommand{\cI}{\mathcal{I}}
\newcommand{\cR}{\mathcal{R}}
\newcommand{\cS}{\mathcal{S}}

\newcommand{\eat}[1]{}

\renewcommand{\epsilon}{\varepsilon}

\renewcommand{\emptyset}{\varnothing}

\newcommand{\MaxCvg}{\textsc{MaxCov}\xspace}
\newcommand{\SetC}{\textsc{SetCov}\xspace}

\DeclareMathOperator{\opt}{OPT}

\newtheorem{result}{Result}

\begin{document}

\title{Improved Algorithms for Maximum Coverage\\ in Dynamic and Random Order Streams} 

\author{%
Amit Chakrabarti%
\thanks{Department of Computer Science, Dartmouth College, Hanover, NH, USA. Supported in part by NSF under award 2006589.}
\and
Andrew McGregor%
\thanks{College of Information and Computer Sciences, University of Massachusetts, Amherst, MA, USA}
\and
Anthony Wirth%
\thanks{School of Computing and Information Systems, The University of Melbourne, Victoria, Australia. Supported in part by the Faculty of Engineering and Information Technology.}
}

\date{}

\maketitle

\begin{abstract}

The maximum coverage problem is to select $k$ sets from a collection of sets such that
the cardinality of the union of the selected sets is maximized. We consider
$(1-1/e-\epsilon)$-approximation algorithms for this NP-hard problem in three
standard data stream models. 
\begin{enumerate}
    \item {\em Dynamic Model.} The stream consists of a
    sequence of sets being inserted and deleted. Our multi-pass algorithm uses
    $\epsilon^{-2} k \cdot \polylog(n,m)$ space. The best  previous
    result (Assadi and Khanna, SODA 2018) used $(n +\epsilon^{-4} k)
    \polylog(n,m)$ space. While both algorithms use $O(\epsilon^{-1} \log n)$
    passes, our analysis shows that when $\epsilon$ is a constant, it is possible to reduce the
    number of passes by a $1/\log \log n$ factor
    without incurring additional space.

    \item {\em Random Order Model.} In this model, there are no deletions and the sets forming the instance are uniformly randomly permuted to form the input stream. We show that a single
    pass and $k \polylog(n,m)$ space suffices for arbitrary small
    constant $\epsilon$.  The best previous result, by Warneke et al.~(ESA
    2023), used $k^2 \polylog(n,m)$ space.

    \item {\em Insert-Only Model.} Lastly, our results, along with numerous
    previous results, use a sub-sampling technique introduced by McGregor and
    Vu (ICDT 2017) to sparsify the input instance. We explain how this
    technique and others used in the paper can be implemented such that the
    amortized update time of our algorithm is polylogarithmic. This also
    implies an improvement of the state-of-the-art insert only algorithms in
    terms of the update time: $\polylog(m,n)$ update time suffices whereas the
    best previous result by Jaud et al.~(SEA 2023) required update time that
    was linear in $k$. 
\end{enumerate}
\end{abstract}

\section{Introduction}
\label{sec:intro}
\paragraph{Background.} The input to the \emph{maximum coverage problem} is an
integer $k$ and a collection of $m$ sets $S_1, \ldots, S_m$, each a subset of
the \emph{universe} $[n] := \{1,\ldots,n\}$. The goal is to find $k$ sets
$\{S_i\}$ whose union has maximum cardinality. This longstanding problem has
several applications, including facility and sensor allocation
\cite{KrauseG07}, circuit layout and job
scheduling~\cite{hochbaum_analysis_1998}, information
retrieval~\cite{anagnostopoulos_stochastic_2015}, influence maximization in
marketing strategy design \cite{KempeKT15}, and content
recommendation~\cite{SahaG09}.  In terms of theoretical importance,
it is perhaps simplest example of a submodular maximization problem, a rich
family of problems in machine learning~\cite{krause_golovin_2014}. A natural
variant of the problem is the \emph{set cover problem} and was one of Karp's
original~21 NP-complete problems~\cite{karp1972reducibility}. A greedy
approach that iteratively chooses the set with the greatest contribution
yields asymptotically optimal approximation algorithms for set cover and
maximum coverage~\cite{Feige98}.  Indeed, the greedy algorithm in practice
typically outperforms its proven approximation factor for many practical
problem instances~\cite{lim2014lazy}.

\paragraph{Data Stream Computation.} Over the last decade, there has been a growing body of work that addresses the challenges of solving the maximum coverage problem (henceforth, \MaxCvg), the related set cover problem (\SetC), and the general problem of submodular optimization, in computational models such as the data stream model, that are relevant when the underlying data sets are massive  \cite{SahaG09, McGregorV19, BateniEM17, indyk_tight_2019,AssadiK18,WarnekeCW23,AusielloBGLP12,YuY13,LiuRVZ21,AgrawalSS19,Norouzi-FardTMZ18,CormodeKW10,JaudWC23,Har-PeledIMV16,ChakrabartiW16,EmekR14,AssadiKL16}. Assadi, Khanna and Li~\cite{AssadiKL16} have a comprehensive summary of results and discussion.
In many real-world applications, even the conceptually simple greedy approach is not practical. This can occur if the underlying data is too large to be stored in a single location and is therefore distributed across multiple machines; or if the data changes over time, e.g., sets are added or removed from the input collection.

This work focuses on the set streaming model~\cite{SahaG09}, where the input stream is a sequence of sets (in the dynamic case, a sequence of set insertions and deletions, arbitrarily intermixed). Each set is described contiguously. The goal is to solve \MaxCvg in space that is sublinear in the size of the input.
We note in passing that there has also been interest in an alternative ``edge streaming'' model~\cite{Khanna0A23,BateniEM17,indyk_tight_2019}, but it is not the focus of this paper.

Along with the set-streaming model itself, Saha and Getoor~\cite{SahaG09} introduced the first algorithm for \MaxCvg in this model. Their algorithm, \textsc{sops},\footnote{
    The algorithm received this name in a later work by Yu and Yuan~\cite{YuY13}
} achieves an approximation factor of $1/4$ using $O(nk)$ space.
The input could be as large as~$\Omega(m n)$, and there is, e.g., a logarithmic-factor approximation for set cover that involves space~$\tilde{O}(mn^\delta)$ in~$O(1/\delta)$ passes~\cite{Har-PeledIMV16}.
To accord with practice, early approaches for set cover took space linear in~$n$~\cite{CormodeKW10,EmekR14,ChakrabartiW16}.
Though it is unclear whether there are sensible~$o(n)$-space solutions for Set Cover, McGregor and Vu~\cite{McGregorV19} pioneered the quest for~$\tilde{O}(k)$-space methods for \MaxCvg.
Storing even a single set from the stream may, in general, require~$\Theta(n)$ space. To avoid this, various papers consider sub-sampling the elements in the sets \cite{JaudWC23,DemaineIMV14,BateniEM17,Har-PeledIMV16}.
 In particular, McGregor and Vu~\cite{McGregorV19} observed that subsampling the instance to a universe of size~$\tilde{O}(k)$ with a hash function of high \emph{degree} preserves coverage very accurately.

We will focus on three specific set streaming models. In the \emph{insert-only set streaming model} the data stream comprises $m$ sets, in  arbitrary order, each described once. In the \emph{random order set streaming model}, we assume that these sets, each described once and contiguously, are ordered uniformly at random. This model captures a natural form of average-case analysis: the sets are chosen adversarially, but we can choose a random ordering to process them~\cite{AssadiK18,
WarnekeCW23}.
In the \emph{dynamic set streaming model}, each set may be added and deleted several times: the resulting \MaxCvg instance involves only those sets that remain added at the end of the stream.

\paragraph{Submodular maximization.}
Our problem \MaxCvg is a special case of \emph{submodular maximization},
specifically, of the problem of maximizing a monotone submodular function
under a cardinality constraint. In the streaming
setting~\cite{badanidiyuru_streaming_2014, feldman_one-way_2020}, we have
oracle access to a submodular function~$f \colon 2^U \to \RR$, for some
universe $U$ and the input stream is a sequence of elements of $U$. These
elements correspond to sets $S_i$ in \MaxCvg. An important point here is that 
in \MaxCvg (unlike in submodular maximization), storing an element of $U$
(i.e., a set $S_i$) takes non-constant space, as does storing the $f$-value
(coverage) of a putative solution take space.

The \textsc{sieve-streaming} algorithm~\cite{badanidiyuru_streaming_2014} for
submodular maximization gives a $(1/2 - \varepsilon)$-approximation for any
$\varepsilon > 0$. By adapting to maximum coverage, there is an
$\tilde{O}(n\varepsilon^{-1})$-space algorithm that achieves the same
approximation factor~\cite{mcgregor_better_2018}.  Subsequently, Chekuri \etal
\cite{ChekuriGQ15} extended this work to non-monotone submodular function
maximization under constraints beyond cardinality.  Observe that
$\Omega(mk^{-3})$ space is required to achieve an approximation greater than
$1/2 + o(1)$ in the arbitrary-arrival model~\cite{feldman_one-way_2020}.
Needless to say, there is progress on submodular maximization in dynamic
streams~\cite{NEURIPS2020_6fbd841e}, though comparatively little about
\MaxCvg.

\paragraph{Random order.}
Streaming algorithms for submodular maximization in the random-arrival model
have been studied~\cite{norouzi-fard_beyond_2018, AgrawalSS19,
liu_cardinality_2021}.  Interestingly, approximation factors better than~$1/2$
have been achieved in this model using low space.  Norouzi-Fard et al.
developed \textsc{salsa}, an $\tilde{O}(k)$-space algorithm that achieves an
approximation factor of $1/2 + c_0$ in-expectation, where $c_0 > 0$ is a small
absolute constant.  Subsequently, Agrawal et al.~\cite{AgrawalSS19} achieved
an approximation factor of $1 - 1 / e - \varepsilon - o(1)$ in-expectation
using $\tilde{O}_\varepsilon(k)$ space.  The $o(1)$ term in the approximation
factor is a function of $k$. For small $k$, this term is very large, in which
case \textsc{salsa} makes a better approximation guarantee; nonetheless the
algorithm of Agrawal et al. is a focus of this paper.  Warneke et
al.~\cite{WarnekeCW23} introduced low-space methods for maximum coverage in
random-order streams. Their approaches consume space proportional to~$k^2$: we
improve this to linear in~$k$, see~\Cref{sec:results} for details.  The
corresponding hardness result: $\Omega(m)$ space is required to achieve an
approximation factor of $1 - 1 / e + \varepsilon + o(1)$ in the random-arrival
model~\cite{liu_cardinality_2021}.\footnote{ This improves upon the
$\Omega(mk^{-3})$ space lower bound given by McGregor and
Vu~\cite{mcgregor_better_2018} for random-arrival set-streaming maximum
coverage, which applies to random-arrival submodular maximization, due to an
important connection between these two problems, explained shortly.  }

\subsection{Our Results, Approach, and Related Work}
\label{sec:results}

To establish our new results, we revisit recent work on the maximum coverage problem in the relevant data stream models and show how to refine and extend the approaches to substantially decrease the the amount of space needed to solve the problem in the data stream model. In the process, we also establish secondary results on speeding up the existing algorithms and in one case, developing a slightly more pass efficient algorithm.

In the dynamic stream model, our main result is the following.

\begin{result}[Algorithm for Dynamic Model]
There exists an $O(1+\epsilon^{-1}/\log \log m) \log m)$ pass algorithm in the dynamic data stream model that uses $\epsilon^{-2} \cdot k \cdot \polylog(n,m)$ space and returns a $1-1/e-\epsilon$ approximation.
\end{result}

This result reappears as \Cref{thm:dyn}. %
It is an improvement over the best previous algorithm by Assadi and Khanna
\cite{AssadiK18} that used $(n+\epsilon^{-4} k) \polylog(n,m)$ space. The
high-level approach is similar and is based on the existing $\ell_0$ sampling
primitive that allows uniform sampling amongst the sets that are inserted but
not deleted. In each pass, we pick additional sets and maintain the union $C$
of the selected sets. We pick the additional sets by sampling only among the
sets that include a ``meaningful'' number of elements that are not currently
in $C$. Unfortunately, when we simultaneously sample multiple sets in this
way, once we start adding sampled sets to our solution, some of the sampled
sets may no longer make a meaningful increment to $|C|$. However, by repeating
the process over multiple passes, with a slowly decreasing threshold for what
constitutes a meaningful improvement, we ensure that either we exhaust the
collection of sets that could potentially be worth adding or we add enough
sets and together  these yield a good approximate solution. The main point of
departure from Assadi and Khanna \cite{AssadiK18} is in the way we handle the
decreasing threshold. The largest threshold considered by Assadi and Khanna is
about $\opt/k$; this is sufficient for achieving a $1-1/e-\epsilon$
approximation factor, but has the downside that there can be many overlapping
sets that would contribute more than $\opt/k$ new elements to $C$; storing
these sets during the sampling process is expensive. Instead, we consider
thresholds decreasing from $\opt$ but handling the sets that make a
contribution between $\opt/2^i$ and $\opt/2^{i+1}$ for $i=0,\ldots, \log k$
separately but in parallel. We analyze a stochastic process---we call it the
\emph{cascading urns} process---to show how our new algorithm uses a similar number of passes as the algorithm by Assadi and Khanna \cite{AssadiK18}
but uses significantly less space. We present the details in
\Cref{sec:dynamic}. In fact, a more careful analysis of the process shows that
it is possible to reduce the number of passes to $O(\log m+ \epsilon^{-1} \log
m /\log \log m)$ in contrast to the $O(\epsilon^{-1} \log m)$ passes used by
previous algorithm. We believe that a similar pass saving can be achieved in
Assadi and Khanna \cite{AssadiK18} but it requires an extra log factor in the
space.

In the random order model, our main result is the following.

\begin{result}[Algorithm for Random Order Model]
There exists a single pass algorithm in the random order data stream model that uses $O_\epsilon(k \cdot \polylog(n,m))$ space and returns a $1-1/e-\epsilon$ approximation of maximum coverage in expectation. 
\end{result}

The above result is established in \Cref{sec:random-order-alg}. Note that the dependence on $\epsilon$ in the result is exponential and
this makes the algorithm less practical. The main significance of the result
is removing a factor $k$ in the space required by the best previous result, by
Warneke et al. \cite{WarnekeCW23}, that used $\epsilon^{-2} k^2 \polylog(n,m)$
space. Both their and our approach are based on modifying existing algorithms
for the cardinality constrained monotone submodular optimization problem. This
is a more general problem than maximum coverage, but it is assumed in the
general problem that the algorithm has oracle access to the function being
optimized. For maximum coverage, we need to store enough information to be
able to simulate this oracle. An $O_\epsilon(k^2 \log m)$-space algorithm for
maximum coverage follows immediately for any $O_\epsilon(k)$ space algorithm
for monotone submodular optimization because the universe subsampling
technique discussed in Section~\ref{sec:prelim} allows us to focus on the case
where set has size $O(\epsilon^{-2} k \log m)$.  To reduce the dependence on
$k$ to  linear, rather than quadratic, we need a more careful adaption of an
algorithm by Agrawal et al.~\cite{AgrawalSS19}. We maintain a set of size
$O(\epsilon^{-2} k \log m)$ corresponding to the items covered by a collection
of sets chosen for the solution so far, along with $O_\epsilon(1)$ sets of
size $O_\epsilon(\log m)$. We mention in passing that, being a thresholding
algorithm, \textsc{salsa} is easy to adapt to our $O_\epsilon(k \cdot
\polylog(n,m))$ regime, yielding an approximation factor of~$1/2+c_0$. We
present the details in \Cref{sec:random}.

Lastly, our results, along with numerous previous results, use the aforementioned universe sub-sampling technique introduced by McGregor and Vu \cite{McGregorV19} to sparsify the input instance. We explain how this technique and others used in the paper can be implemented such that the amortized update time of our algorithm is polylogarithmic. This also implies an improvement of the state-of-the-art insert only algorithms in terms of the update time: $\polylog(m,n)$ update time suffices whereas the best previous result by Jaud et al.~\cite{JaudWC23} required update time  $k\cdot \polylog(m,n)$. We present the details in \Cref{sec:fast}.

\section{Further Related Work}
Here we provide some more detail on the background materials.

\paragraph{Random arrivals.}
Guha and McGregor~\cite{guha_stream_2009} \S 1.1 further justify the assumptions of the random-arrival model.
This model has been studied for many important problems in other contexts, including quantile estimation~\cite{guha_stream_2009}, submodular maximization~\cite{Norouzi-FardTMZ18, AgrawalSS19}, and maximum matching~\cite{konrad_maximum_2014, assadi_beating_2021}, often revealing improved space-accuracy trade-offs.

\paragraph{Coverage in streams.}
Saha and Getoor \cite{SahaG09} gave a swap-based $1/4$ approximation  algorithm, \textsc{sops}, that uses a single pass and $\tO(kn)$ space. Their algorithm stores $k$ sets explicitly in the memory as the current solution and selectively replaces sets in memory as new sets arrive in the stream. 
Subsequently, Ausiello \etal \cite{AusielloBGLP12} gave a slightly different swap based algorithm that also finds a $1/4$ approximation using one pass and the same space.

Yu and Yuan~\cite{YuY13} improved upon \textsc{sops} under a more relaxed output specification, in which only the \textit{IDs} of the chosen sets must be returned.
Their algorithm algorithm, \textsc{gops}, achieves an approximation factor of approximately~0.3 using just $\tilde{O}(n)$ space.

 McGregor and Vu~\cite{McGregorV19} introduced several methods for \MaxCvg
 under set streaming. They presented four near-optimal algorithms for the
 problem, with the following specifications:
\begin{enumerate}
    \item An approximation factor of $1 - 1 / e - \varepsilon$ using $\tilde{O}(m\varepsilon^{-2})$ space.
    \item An approximation factor of $1 - 1 / e - \varepsilon$ using $\tilde{O}(k\varepsilon^{-2})$, requiring $\tilde{O}(\varepsilon^{-1})$ passes over the stream. This method is a key launchpad for our work.
    \item An approximation factor of $1 - \varepsilon$ using $\tilde{O}(m\varepsilon^{-3})$ space (and exponential time).
    \item An approximation factor of $1/2 - \varepsilon$ using just $\tilde{O}(k\varepsilon^{-3})$ space.
\end{enumerate}
Subsequently, McGregor et al.~\cite{McGregorTV21} showed how to solve max coverage exactly in one pass and in space~$\tilde{O}(d^{d+1}k^d)$ when every set is of size at most~$d$.

\paragraph{Dynamic streams.}
In the dynamic graph stream model~\cite{McGregor14},
the stream consists of insertions and deletions of edges of the underlying graph~\cite{AhnCGMW15,AhnGM12a,AhnGM12b,AhnGM13,KapralovLMMS14,KapralovW14,GuhaMT15,BhattacharyaHNT15,ChitnisCEHMMV16,AssadiKLY16,Konrad15,McGregorVV16,McGregorTVV15}, relevant to the maximum vertex coverage problem. An algorithm for the dynamic graph stream model can also be used in the streaming-set model; the streaming-set model is simply a special case in which edges are grouped by endpoint.

\paragraph{Lower bounds.}
McGregor and Vu~\cite{McGregorV19} introduced the first nontrivial space lower
bound for set-streaming maximum coverage. For all $\varepsilon > 0$, achieving
an approximation factor better than $1 - 1 / e + o(1)$ requires
$\Omega(mk^{-2})$ space in the arbitrary-arrival case, and $\Omega(mk^{-3})$
space in the random-arrival case, where the $o(1)$ term is a function of $k$.
This also represented the first result explicitly regarding random-arrival
maximum coverage.  Assadi~\cite{assadi_tight_2017} showed that achieving a $(1
- \varepsilon)$-approximation requires $\Omega(m\varepsilon^{-2})$ space, even
in the random-arrival model, almost matching the
$\tilde{O}(m\varepsilon^{-3})$ upper bound achieved by McGregor and
Vu~\cite{McGregorV19}. More recently, Feldman et
al.~\cite{feldman_one-way_2020} showed that $\Omega(mk^{-3})$ space is
required to beat $1/2 + o(1)$ in the arbitrary-arrival model.

\section{Preliminaries}
\label{sec:prelim}
Without loss of generality we assume each set is tagged with a unique ID in the range $[m^3n]$, for if not, we could generate suitable IDs via random hashing of the contents of the sets themselves.\footnote{E.g., hash each set $S\in [n]$, to $f_S(r)$, where $f_S(X)=\prod_{u\in S} (X-u)$ is a polynomial over a prime field $\FF_p$ of cardinality $\Theta(m^3 n)$ and $r$ is uniformly random in $\FF_p$. By elementary analysis, the map $S \mapsto f_S(r)$ is non-injective with probability at most $\binom{m}{2}n/p\leq 1/m$.} In the dynamic setting, we will want to sample uniformly from the sets that are inserted but not deleted (sometimes we will put additional requirements on the relevant sets). To do this, we will use the standard technique of $\ell_0$ sampling.

\begin{theorem}[$\ell_0$ sampling \cite{JowhariST11}]
  There is a one-pass algorithm that processes a stream of tokens $\langle
  x_1, \Delta_1\rangle, \langle x_2, \Delta_2\rangle, \ldots, $ where each
  $x_i\in \{1,\ldots, M\}$ and $\Delta_i\in \{-1,1\}$, using $O(\log^2(M)
  \log(1/\delta))$ bits of space, that, with probability $1-\delta$, returns
  an element chosen uniformly at random from the set $\{x\in [M]:
  \sum_{i:x_i=x} \Delta_i\neq 0\}$.
\end{theorem}
Specifically, we will use the above technique when $M = m^3 n$, with each
$x_i$ being the ID of a set and the corresponding $\Delta_i$ specifying
whether the set is being inserted or deleted. 

As mentioned in the introduction, the natural (offline) greedy approach achieves a
$(1-1/e)$-approximation for maximum coverage. In streaming settings, a
``quantized'' version of the algorithm achieves similar results, as summarized
below.

\begin{theorem}[Quantized Greedy Algorithm, e.g.,~\cite{McGregorV19}]\label{thm:quantizedgreedy}
  Consider the algorithm that processes a stream $\langle S_1, S_2, \ldots
  \rangle$ of sets as follows. Starting with an empty collection $Y$, it makes
  $p$ passes: in the $i$th pass, it adds to $Y$ each encountered set that
  covers $\ge \tau_i$ uncovered elements. If, at some point, $|Y| = k$, then
  the algorithm stops and returns $Y$. If these threshold parameters $\tau_1 >
  \cdots > \tau_p$ satisfy 
  \begin{enumerate}
    \item $\tau_1\geq \opt/k$,
    \item $\tau_p<\opt/(4ek)$, and
    \item $\tau_{i}/\tau_{i+1}\leq 1+\epsilon$, for all $i\in [p-1]$,
  \end{enumerate} 
  then the algorithm achieves a $(1-1/e-\epsilon)$-approximation for the
  maximum coverage problem. 
\end{theorem}

Both of our main algorithmic results appeal to the following technique that,
given a guess\footnote{We run the final algorithms with guesses $v=1, 2, 4, 8,
\ldots 2^{\log n}$  and use the fact that one of these guesses is within a
factor 2 of $\opt$.} $v$ of the value $\opt$, transforms the \MaxCvg instance
into one in which the optimum solution covers $O(\epsilon^{-2} k \log m)$
elements. Note that this immediately implies that every set in the input has
cardinality $O(\epsilon^{-2} k \log m)$. In what follows, we will tacitly
assume that the given instance is filtered through this technique before being
fed into the algorithms we design. In particular, we will appeal to the following theorem \cite{McGregorV19,JaudWC23}.

\begin{theorem}[Universe Subsampling]
  Let the function $h:[n]\rightarrow \{0,1\}$ be drawn uniformly at random
  from a family of $O(k \log m)$-wise independent hash functions, where 
  \[
    p := \prob{h(e) = 1} = {\lambda}/{v}, \quad \text{for } \lambda = 10 \epsilon^{-2} k \,,
  \]
  and $v$ satisfies $\opt/2 \leq v \leq \opt$.  An $h$-sparsification of an
  instance $\cI$ of \MaxCvg is formed by replacing every set $S$ in $\cI$
  by $\{e\in S: h(e)=1\}$.
  With high probability, any $\alpha$-approximation for the $h$-sparsification
  of $\cI$ yields an $(\alpha-\epsilon)$-approximation for $\cI$.
\end{theorem}

In Section~\ref{sec:fast}, we will discuss how to implement the universe subsampling such that the update time of the resulting algorithm is logarithmic for each element of a set in the stream.

\section{Dynamic Streams}
\label{sec:dynamic}
Our main algorithm for dynamic streams is based on the following variant
(\Cref{alg:dyn}) of the quantized greedy algorithm. It requires an
estimate $v$ satisfying $\opt/2 \leq v < \opt$. To describe it compactly, we
set up some notation. Let $\cS$ denote the collection of sets yielded by the
input stream. For a given $C\subset [n]$, define the subcollections
\begin{alignat*}{2}
  \cF_i^C &= \{S \in \cS:  \theta_i \leq |S\setminus C| < \theta_{i-1}\} \,, 
    &\quad& \text{for } i \in [\ell] \,, \\
  \cG_{i}^C &= \{S\in \cS: \tau_i \leq |S\setminus C| < \tau_{i-1} \} \,,
    && \text{for } i \in [1+\lceil \log_{1+\epsilon} (16e)] \rceil \,,
\end{alignat*}
where
\[
  \ell = \rounddown{\log k} \,, \quad 
  \theta_i :=\frac{2v}{2^i} \,, \quad \text{and~~}
  \tau_i := \frac{2v/2^{\ell}}{(1+\epsilon)^{i-1}} \,.
\]
In the pseudocode below, the computed solution $Y$ is a set of IDs and the set
$C$ is maintained to be the subset of the universe covered by the sets
represented in $Y$. The macro\footnote{Being a macro, \textsc{grow-solution}
can cause the invoking algorithm to return.}
\Call{grow-solution}{$Y,S,C,\theta$} implements the following logic: if $|S
\setminus C| \ge \theta$, then update $Y \gets Y \cup \{\id(S)\}$ (i.e., add
the set $S$ to the solution) and $C \gets C \cup S$; if this makes $|Y| = k$,
then stop and return $Y$.
\begin{algorithm}[!htbp]
\caption{Quantized greedy adapted for dynamic set streams}
\begin{algorithmic}[1]
\State start with an empty solution ($Y \gets \emptyset$) and let $C \gets \emptyset$
\While{$\cF_1^C\cup \cdots \cup \cF_\ell^C \ne \emptyset$} \label{line:cascade}
    \For{$i \in [\ell]$} {sample $2^i$ sets from $\cF_i^C$ with replacement} \EndFor
    \For{$i\in [\ell]$ and each set $S$ sampled from $\cF_i^C$}
        \State \Call{grow-solution}{$Y,S,C,\theta_i$}
    \EndFor
\EndWhile
\For{$i \gets 1$ to $1 + \roundup{\log_{1+\epsilon}(16e)}$}
    \While{$\cG_{i}^C \ne \emptyset$} \label{line:one-urn}
        \State sample $k$ sets from $\cG_{i}^C$ with replacement
        \For{each sampled set $S$} {\Call{grow-solution}{$Y,S,C,\tau_i$}} \EndFor
    \EndWhile
\EndFor
\end{algorithmic}
\label{alg:dyn}
\end{algorithm}

The most involved part of the analysis of \Cref{alg:dyn} is
the proof of the following lemma, which we shall defer slightly.

\begin{lemma} \label{lem:dyn-iters}
  With high probability, the while loop at \cref{line:cascade} makes
  $O(\log m)$ iterations and each invocation of the while loop at
  \cref{line:one-urn} makes $O(\log m/\log \log m)$ iterations.
\end{lemma}

\begin{theorem} \label{thm:dyn}
  There is a $(1-1/e-\epsilon)$-approximation algorithm for max coverage in
  the dynamic stream model that uses $\tilde{O}(k/\epsilon^{2})$ space and 
  $O((1+\epsilon^{-1}/\log \log m) \log m)$ passes.
\end{theorem}
\begin{proof}
  The idea is to implement \Cref{alg:dyn} in a small number of space-efficient
  streaming passes. To sample from $\cF_{i}^C$ and check the non-emptiness
  condition of the while loop at \cref{line:cascade}, we use an
  $\ell_0$-sampling pass on the IDs of the sets where $|S\setminus C|$ is in
  the relevant range.  Using an additional pass, we store $\{S\setminus C\}$
  for all the sampled sets, so that we can implement the logic of
  \textsc{grow-solution}. These sets contain at most
  \[
    \sum_{i=1}^{\ell} 2^i \cdot 2v/2^{i-1} 
    = 4v\ell = O(\epsilon^{-2} k \log k)
  \] 
  elements. Similarly we can sample sets from each $\cG_{i}^C$ while only storing
  $O(\epsilon^{-2} k \log k)$ elements. 

  Since each iteration of each while loop
  corresponds to two streaming passes, 
  \Cref{lem:dyn-iters} shows that
  w.h.p.~the algorithm uses at most the claimed number of passes.
  The approximation guarantee follows from \Cref{thm:quantizedgreedy},
  since 
  \[
    \tau_1=2v/2^{\ell} > \opt/k
  \]
  and
  \[
    \tau_{1 + \roundup{ \log_{1+\epsilon} (16e) }}
    = \frac{2v/2^{\ell}}{(1+\epsilon)^{\roundup{ \log_{1+\epsilon} (16e) }}}
    < \frac{4\opt}{16ek}=\frac{\opt}{4ek} \,.
    \qedhere
  \] 
\end{proof}
 
It remains to prove \Cref{lem:dyn-iters}. We need to understand how the
collections $\cF_i^C$ and $\cG_i^C$ evolve as we grow our solution $Y$, thus
changing $C$. To this end, we introduce two stochastic {\em urn processes}:
the first (and simpler) process models the evolution of the collection
$\cG_i^C$, for a particular $i$; the second process models the evolution of
the ensemble $\{\cF_i^C\}$. How exactly the urn processes model these evolutions is explained in \Cref{sec:urn-applic}.

\subsection{The Single Urn Process and its Analysis}

An urn contains $m$ balls, each initially {\em gold}. Balls are drawn from the
urn in phases, where a phase consists of drawing $d$ balls uniformly at
random, with replacement. When a ball is drawn, if it is gold, then we earn
one point and some arbitrary subset of the balls in the urn turn into {\em
lead}, including the ball that was just drawn (and put back). At the end of
the phase, all lead balls are removed from the urn and the next phase begins.
The process ends when either $d$ points have been earned or the urn is empty.

We will prove the following result about this process.

\begin{theorem} \label{thm:urnrounds} 
  With probability $\ge 9/10$, the urn process ends within $O(\log m / \log
  \log m)$ phases. 
\end{theorem}
%
The intuition behind the result is as follow. The observation is that if the fraction of balls in the urn remains above $d'/d$ then each draw is gold with probability at least $d'/d$ and it would be reasonable to expect at least $d'/d \times d=d'$ of the draws to be gold.
Define $m(i)$ to be number of gold balls left in the urn after $i-1$ rounds and let  $d(i)$ be the number of gold balls observed in the $i$-th round. 
For the sake of intuition (the formal proof follows shortly) suppose the observation holds for all rounds even though $d(i)$ and $m(i+1)$ are dependent, i.e.,
$m(i+1)/m(i) \leq  d(i)/d$ for all~$i$;
 then 
\[\frac{m(i+1)}{m} = \prod_{j=1}^{i+1}\frac{m(j)}{m(j-1) } \leq \prod_{j=1}^{i}\frac{d(j)}{d} \leq  \left (\frac{(\sum_{j=1}^i d(j)/d)}{i} \right )^i\leq 1/i^i\,,\] by the AM-GM inequality and the fact that every drawn gold ball is turned into lead, so only contributes to at most one term~$d(j)$, hence $(\sum_j d(j)/d)\leq 1$.
Consequently, $i$ can be at most $O(\log m/\log \log m)$ before $m(i+1)$ becomes zero. 

Our result is stronger than analogous results in previous work in two specific ways. Previous analyses essentially need a) an factor $O(\log m)$ in the number of draws taken in each phase and b) only establish that $O(\log m)$ phases suffice, rather than $O(\log m/\log \log m)$ phases. The improvements in our approach stem from avoiding the need to guarantee that each phase makes progress with high probability and considering both the progress made towards observing $d$ gold balls and the progress made in terms of turning balls into lead.

\begin{proof}[Proof of \Cref{thm:urnrounds}]
  Let $m_j$ be the number of balls in the urn after the $j$th phase ends and
  all lead balls are removed. Also, let $m_0 = m$. Put
  \begin{align} \label{eq:gamma-def}
    \gamma = \log \log m/\log m
  \end{align}
  and assume $d\geq 12/\gamma$, otherwise the result follows trivially, since
  each phase earns at least one point. The $j$th phase is deemed
  \emph{successful} if, during it, either $\gamma d/2$ points are earned or
  the fraction of gold balls drops below $\gamma$ (causing $m(j)/m(j-1) <
  \gamma$). In \Cref{lem:phasesuc1,lem:phasesuc2}, we will establish that each
  phase is successful with probability at least $1/2$, even when conditioned
  on the outcomes of previous phases. Thus, by a Chernoff bound, after
  $10/\gamma$ phases the probability that there were at least $4/\gamma$
  successful phases is at least $1-\exp(-O(1/\gamma))$.  If $4/\gamma$ phases
  are successful, either we have earned $(2/\gamma)(\gamma d/2)$ points or the
  number of balls in the urn has reduced from $m(0) = m$ to at most
  $(\gamma)^{2/\gamma} m$. By \cref{eq:gamma-def}, this number is less than
  $1$, implying that the urn is empty.
\end{proof}

We turn to lower bounding the success probability of a phase as claimed in the
above proof. Fix a particular phase and a particular realization of all
previous phases of the stochastic process. Let $G_i$ be the event that the
$i$th draw of this phase reveals a gold ball (and thus earns a point); let
$L_i$ be the event that after $i-1$ draws, the fraction of gold balls in the urn 
is less than $\gamma$. Define indicator variables $X_i = \mathbbm{1}_{G_i \cup
L_i}$ and put $X=\sum_{i=1}^d X_i$. Then,
\begin{align} \label{eq:xi}
  \EE[X_i] = \Pr[L_i] + \Pr[G_i \cap \overline{L}_i] 
  \ge \Pr[G_i \mid \overline{L}_i] 
  \ge \gamma \,.
\end{align}

\begin{lemma}\label{lem:phasesuc1}
  If $X \ge \gamma d/2$, then the phase is successful.
\end{lemma}
\begin{proof}
  If $X \ge \gamma d/2$, either some event $L_i$ occurs or we collect $\gamma
  d/2$ gold balls (and points). In the latter case, the phase is successful by
  definition. In the former case, the fraction of gold balls drops below
  $\gamma$ during the phase. Since lead never turns into gold, this fraction
  remains below $\gamma$. Upon removing lead balls at the end of the phase,
  the number of balls in the urn drops to at most $\gamma$ times the number at
  the start of the phase, i.e., the phase is successful.
\end{proof}

Finally, we show that $X$ stochastically dominates a binomial random variable,
which then lower-bounds $\Pr[X\geq \gamma d/2]$. Denote Bernoulli and binomial
distributions as ``Bern'' and ``Bin.''

\begin{lemma}\label{lem:phasesuc2}
  $\Pr[ X\geq \gamma d/2] \geq 1-\exp(-\gamma d/12) \geq 1/2$, where the latter
  inequality uses our assumption that $d\geq 12/\gamma$.
\end{lemma}
\begin{proof}
  Let $Y_1, \ldots, Y_d$ be independent draws from $\Bern(\gamma)$. Put
  $X^{\leq j}=\sum_{i=1}^j X_i$ and $Y^{\leq j}=\sum_{i=1}^j Y_i$. We first
  show by induction on $j$ that $X^{\leq j}$ stochastically dominates $Y^{\leq
  j}$ for all $j$. Inequality~\eqref{eq:xi} with $i=1$ establishes the base
  case of $j=1$. Assuming the result for $j$, for an arbitrary integer, $z$,
  we have
  \begin{align*}
  & \Pr[X^{\leq j+1}\geq z] \\
  &=   \Pr[X^{\leq j}\geq z] + \Pr[X^{\leq j}= z-1] \Pr[X_{j+1}=1 \mid X^{\leq j}= z-1] \\
  &\ge \Pr[X^{\leq j}\geq z] + \Pr[X^{\leq j}= z-1] \gamma \\
  &=   \Pr[X^{\leq j}\geq z] (1-\gamma) + \Pr[X^{\leq j}\geq  z-1] \gamma \\
  &\ge \Pr[Y^{\leq j}\geq z] (1-\gamma) + \Pr[Y^{\leq j}\geq  z-1] \gamma \\
  &=   \Pr[Y^{\leq j+1}\geq z] \ .
  \end{align*}
  Thus, $\Pr[ X\geq \gamma d/2]\geq \Pr[ Y\geq \gamma d/2]$. The lemma now
  follows by applying a Chernoff bound to $Y \sim \Bin(d,\gamma)$.
\end{proof}

\subsection{Cascading Urns}

To prove the first part of Lemma \ref{lem:dyn-iters}, we will also need to analyze a more complex version of the above process; we
call it the {\em cascading urn process}.  Now we have~$t$ urns, numbered $1$
through $t$, with the $r$th urn starting out with $m_r$ balls; let $m =
\sum_{r\in [t]} m_r$.  Drawing a gold ball from urn~$r$ is worth~$d/2^r$
points. Initially, all balls are gold. The process works as follows.
\begin{itemize}
  \item Balls are drawn in phases: a phase consists of drawing $2^r$ balls
  from urn~$r$, in order of increasing $r$. 
  \item When a ball is drawn from urn~$r$, if it is gold, the following things
  happen: (i)~it turns into lead and is returned to urn~$r$; (ii)~some
  arbitrary subset of the balls in the urns (all urns, not just urn~$r$) turn
  into lead; and (iii) we earn $d/2^r$ points.
  \item At the end of each phase, each lead ball in an urn 
  is removed from its current urn and either destroyed or transformed back
  into gold and placed in a higher-numbered urn. We then start the next phase
  with all balls being gold again.
\end{itemize}
The process ends when either $d$ points have been earned or all urns have
become empty.

\begin{theorem}\label{thm:cascaderounds}
  With probability at least $9/10$, the cascading urn process ends within $O(t
  + \log m)$ phases.
\end{theorem}
\begin{proof}%
  Fix a particular phase. For each $r \in [t]$, the analysis of one phase of
  the single urn process applied to urn~$j$ establishes that, with probability
  at least 
  \[
    \Pr\left[\Bin\left(2^r, \frac12\right) \geq \frac{2^r}{4}\right]
    \geq 1-\exp \left (\frac{-2^{r-1}}{12}\right) 
    \geq 1-e^{-\frac{1}{12}}
    > \frac{1}{13} =: 4\lambda \,,
  \]
  either at least $d/2$ points are earned from this urn (call this event
  $G_r$) or the fraction of gold balls in this urn drops below $1/2$ (call
  this event $L_r$). Urn~$r$ is deemed \emph{successful} in this phase if $G_r
  \cup L_r$ occurs. Put $Z_r = \mathbbm{1}_{G_r \cup L_r}$ and note $\EE[Z_r]
  > 4\lambda$. 

  Let $\hatm_r$ be the number of balls in urn~$r$ at the start of the phase
  we are considering (recall that these are all gold balls) and let $\hatm =
  \sum_{r \in [t]} \hatm_r$.  If $\bigcup_r G_r$ does not occur, then at the
  end of the phase, the number of gold balls across all urns is at most
  \[
    Z := \sum_{r\in [t]} (Z_r \hatm_r/2 + (1-Z_r) \hatm_r) 
    = \sum_{r\in [t]} (1-Z_r/2) \hatm_r \,.
  \]
  Using the lower bound on $\EE[Z_r]$, we get 
  \[
    \EE[Z] \leq \sum_{r\in [t]} (1-2\lambda) \hatm_r = (1-2\lambda)m \,.
  \]
  Applying a Markov bound gives
  \[
    \Pr[Z \leq (1-\lambda)m] \geq 1-\EE[Z]/((1-\lambda)m) = \lambda/(1-\lambda) \,.
  \]
  Hence, with probability at least $\lambda/(1-\lambda)$, either $\bigcup_r
  G_r$ occurs (in which case we collect $d/2$ points) or the fraction of
  balls that remain gold until the end of the phase is at most $1-\lambda$.

  Let $m_{r,p}$ denote the number of balls in urn~$r$ after the $p$th phase
  ends (and all lead balls are either reconverted to gold or destroyed). Also,
  let $m_{r,0} = m_r$. Consider the quantity
  \[ 
    Q_p := \sum_{r \in [t]} 2^{t-r} m_{r,p} \,.
  \]
  Note that $Q_0 \leq 2^{t-1} m$ and if $Q_p < 1$ for some $p$, then the urns
  must be empty after phase~$p$. We claim that if the fraction of balls that
  remain gold in phase $p$ is at most $1-\lambda$, then  $Q_{p}\leq (1 -
  \lambda/2) Q_{p-1} $. This claim holds because (a)~after moving a ball to a
  higher-numbered urn or destroying it, the contribution of that ball to $Q_p$
  goes down by a factor of $2$ or more; and (b)~at least a $\lambda$ fraction
  of the balls turn to lead and are therefore either moved or destroyed. Thus,
  at least a $\lambda$ fraction of the quantity $Q_p$ is reduced to $\lambda
  Q_p/2$ by the end of the phase, proving the claim.

  Hence, with probability at least $\lambda/(1-\lambda)$, in each phase, we
  either collect $d/2$ points or $Q_p$ reduces by a factor $\mu :=
  1-\lambda/2$. By a Chernoff bound, after $p = O(\log Q_0) = O(\mu^{-1} \log
  Q_0)$ phases, either we have collected more than $2\cdot d/2=d$ points or
  $Q_p < 1$, implying that all the urns are empty.
\end{proof}

\subsection{Applications to Our Algorithm} \label{sec:urn-applic}

We now circle back to \Cref{alg:dyn}. To finish its analysis, we need to prove
\Cref{lem:dyn-iters}, which we can now do, as an application of what we have
established about these urn processes.

To analyze the while loop at \cref{line:cascade}, apply
\Cref{thm:cascaderounds} with $t = \rounddown{\log k}$ as follows. The balls
in the $i$th urn correspond to the sets in $\cF_{i}^C$. A ball is {\em gold}
if the corresponding set $S$ satisfies $|S\setminus C| \geq \theta_i$;
otherwise, it is {\em lead}. Adding a set $S$ to the candidate solution $Y$
grows the set $C$ of covered elements, thereby turning gold balls to lead in
some complicated way that the algorithm cannot easily track. The bound on the
number of phases until the urn process terminates translates to a bound of
$O(\log k + \log m)$ on the number of iterations of that while loop. Noting that $k \le m$ gives \Cref{lem:dyn-iters}.

Similarly, \Cref{thm:urnrounds} implies that the number of iterations of the
while loop at \cref{line:one-urn} is $O(\log m/\log \log m)$. This completes
the analysis.

\section{Random Order Model}
\label{sec:random}
In this section, we consider the  stream to be a random permutation of the collection of sets.
We will show that it is possible to approximate the maximum coverage problem up to a factor $1-1/e-\epsilon$ using $O_\epsilon(k \log m)$ space in a single pass. As mentioned in the introduction, our approach is based on modifying an algorithm by Agrawal et al.~\cite{AgrawalSS19} for cardinality constrained monotone submodular maximization. This is a more general problem than maximum coverage, but their algorithm assumes oracle access to the function being optimized. For maximum coverage we need to store enough information to be able to simulate this oracle. A $O_\epsilon(k^2 \log m)$-space algorithm for maximum coverage follows immediately via the universe subsampling technique discussed in Section~\ref{sec:prelim} since we may assume that every set has size $O_\epsilon(k \log m)$.
To reduce the dependence on $k$ to  linear rather than quadratic, we need a more careful adaption of the algorithm in Agrawal et al.~\cite{AgrawalSS19}, to ensure that the algorithm can be implemented via maintaining a set of size $O_\epsilon(k \log m)$ corresponding to the items covered by a collection of sets chosen for the solution so far, along with $O_\epsilon(1)$ sets of size $O_\epsilon(\log m)$.

\subsection{Warm-Up} To motivate the approach, consider a variant of the problem in which an algorithm is presented with a random permutation of the input sets, but we make the following two changes:
\begin{enumerate} 
\item For a fixed optimal solution consisting of sets $O_1, \ldots, O_k$, replace the occurrences of these sets in the stream by a set drawn uniformly at random from $\{O_1, \ldots, O_k\}$.
\item We give the algorithm foreknowledge of a way to segment the stream into $k$ contiguous sub-streams $\cC_1, \ldots, \cC_{k}$ such that each contains exactly one element from $\{O_1, \ldots, O_k\}$. 
\end{enumerate}
Then consider the greedy algorithm that outputs the sequence $\langle S_1, \ldots, S_k\rangle$  where \[S_j = \argmax_{S \in \cC_j} |S \setminus (S_1 \cup \cdots \cup S_{j-1})| \ . \]
To analyze this algorithm, define \[u_j=\opt - |S_1 \cup \cdots \cup S_{j}| ~\mbox{  and }~F_j=\frac{|S_j \setminus (S_1 \cup \cdots \cup S_{j-1})|}{u_{j-1}}\,,\]
and note that 
\[|S_1 \cup \cdots \cup S_k|=\opt - u_k=\opt - \opt \cdot \prod_{j=1}^k ( 1-F_j) \ \]
since for all $j\in [k]$,
\[
u_j=\opt - |S_1 \cup \cdots \cup S_{j-1 }|- |S_j \setminus (S_1 \cup \cdots \cup S_{j-1})|
=(1-F_j) u_{j-1}   \ . \]
Note also that 
\begin{align*}
\expec{F_j \mid S_1, \ldots, S_{j-1}} 
& \ge
 \sum_{i=1}^k \frac{|O_i \setminus (S_1 \cup \cdots \cup S_{j-1})|}{u_{j-1}} \cdot  \Pr[O_i \in \cC_j]\\
&= 
\frac{1}{k} \cdot \sum_{i=1}^k \frac{|O_i \setminus (S_1 \cup \ldots \cup S_{j-1})|}{u_{j-1}} \geq 1/k \ . 
\end{align*}
Hence, \[\EE[|S_1 \cup \ldots \cup S_k|]=\opt(1-\EE[\prod_{j=1}^k ( 1-F_j)])\geq 1-(1-1/k)^k\rightarrow
1-1/e \ .\]
The space required to implement this algorithm is $O(k \epsilon^{-2} \log m)$ since it suffices to maintain the union of the sets chosen thus far rather than the sets themselves.

To handle the original random order setting, we obviously need to deal with the fact that the algorithm does not have foreknowledge of a segmentation of the stream such that each segment contains exactly one set from the optimum solution. A naive approach would be to segment the stream into $\beta k$ contiguous segments for some large constant $\beta>0$ with the hope that few pairs of optimum sets appear in the same segment. We could then consider all $\binom{
\beta k}{k}=e^{O(k)}$ subsequences of these segments but this is clearly inefficient.
A better approach is to use limited number of guesses in parallel using an approach by Agrawal et al.~\cite{AgrawalSS19}. More complicated to analyze, especially when combined with the goal of using limited space, is the fact that because the~$O$-sets appear randomly permuted, rather than independently sampled, there are now various dependencies to contend with.

\subsection{General Setting} \label{sec:random-order-alg}

We will randomly partition the input collection $\cC$ of sets into $k\beta$ groups $\cC_1, \ldots, \cC_{k\beta}$. These groups will ultimately correspond to segments of the stream, i.e., the first $|\cC_1|\sim Bin(m,1/(k\beta))$ sets define $\cC_1$ etc. For any  \[\cI=\{i_1, i_2, \ldots\} \in \{1, \ldots, k\beta\} ~,\] we define a sequence of groups
\[
\Sigma_\cI=\langle \cC_{i_1}, \cC_{i_2}, \cdots , \cC_{i_{|\cI|}} \rangle
~~\mbox{ where $i_1<i_2<\ldots < i_{|\cI|}$} \ . \]

At the heart of the algorithm is  a greedy process run on such sequences. For
$\cI=\{i_1, i_2, \ldots\} \in \{1, \ldots, k\beta\}$, $C \subset [n]$, and a
``reserve'' collection of sets,~$\cR$---the name will become clear later---define the greedy sequence:
\begin{equation} \sigma_\cI(C,\cR)=\langle S_1, S_2, \ldots, S_{|\cI|} \rangle \label{greedyseq}
\end{equation}
where \[S_j=\argmax_{S \in \cR \cup \cC_{i_j}} |S\setminus (C \cup S_1 \cup \cdots \cup S_{j-1})| \ .\]

\begin{algorithm}[!htbp]
\caption{Max Coverage for Random Order Streams}
\label{alg:rand}
\begin{algorithmic}[1]
\State initialize $\cR\leftarrow \emptyset$, $C\leftarrow \emptyset$, $k'\leftarrow 0$.
\For{$i \gets 1$ to $k/\alpha$}
   \State Compute $\sigma_\cI(C,\cR)$ for all $\cI\in P_i$, where $P_i$ consists of all size $k^+ \equiv \min\{\alpha, k-k'\}$ subsets of \[W_i=\{\alpha \beta (i-1)+1, \ldots, \alpha \beta i\} \ .\] We will subsequently refer to $W_i$ as the $i$th \emph{window}.\label{windowline}
   \State Let \[\cI^*=\argmax_{\cI \in P_i} \left  |~ C\cup \left (\bigcup_{S\in  \sigma_{\cI}(C,\cR)} S\right ) ~ \right | \] and update
\[
C\leftarrow C\cup \left (\bigcup_{S\in \sigma_{\cI^*(C,\cR)}} S \right ) 
~~\mbox{ and }~~ k'\leftarrow k'+k^+ \ .
\]
\State For all $\cI\in P_i$, add all sets in $\sigma_\cI(C,\cR)$  to $\cR$.  If there exists a set $S$ in $\cR$ such that $|S\setminus C|\geq \opt/k$, select such a set uniformly at random and  update
\[
C\leftarrow C\cup S ~~\mbox{ and }~~ k'\leftarrow k'+1 \ .
\]
Repeat until either $k'=k$ or no more such sets exist in $\cR$. 
\EndFor
\end{algorithmic}
\end{algorithm}

Algorithm~\ref{alg:rand} contains most details,
but to fully specify the algorithm, we need to define the original groups $\cC_1, \cC_2, \ldots, \cC_{k\beta}$. To do this, for each set $S$ define a random variable $Y_S$ that is uniform over $\{1, 2, \ldots, k\beta\}$. Then, we define $\cC_i=\{S: Y_S=i\}$. Equivalently, for a random permutation of the stream $\cC_1$ is the first $|\{S: Y_S=1\}|$ sets in the stream, $\cC_2$ is the next $|\{S: Y_S=2\}|$ sets in the stream, etc. Given this definition of the stream, Algorithm~\ref{alg:rand}   maintains sets whose total size is \begin{equation}|\opt | + \alpha \binom{\alpha \beta}{\alpha} |\opt | + (k/\alpha) \binom{\alpha \beta}{\alpha} \opt  /k = O_{\alpha,\beta} ( \opt)\,, \label{eq:randsize}\end{equation}
where the first term corresponds to maintaining $C$, the next term consists of all sets appearing in greedy subsequences during the processing of each window (defined in line~\ref{windowline} of Algorithm~\ref{alg:rand}), and the last term corresponds to storing sets in $\cR$.
For sets $S$ in $\cR$ it suffices to store $S\setminus C$ rather than $S$ itself. Hence, the total number of elements in sets in $\cR$ is at most $|\cR| \opt/k$.

The motivation for how to set $\alpha$ and $\beta$ is as follows. Let $\cO=\{O_1, \ldots, O_k\}$ be an optimum collection of $k$ sets.  We say a group  $\cC_i$  is \emph{active} if it contains  a set from $\cO$ and  note if $\beta$ is large enough  the number of active groups $\kappa$  is close to $k$. Furthermore, in each window the number of active groups is expected to be $\kappa \cdot (\alpha \beta)/(k\beta)=\alpha \kappa/k \approx \alpha$. We will later set $\alpha$ to be large enough such that the number of active groups in each window is typically close to $\alpha$, i.e., there is an $\alpha$-length subsequence of the window that mainly consists of active groups. 

\subsubsection{Analysis}

The analysis in Agrawal et al.~\cite{AgrawalSS19} establishes that, during the $i$th window, the algorithm finds a collection of $\alpha$ sets $S_1, \ldots, S_\alpha$ -- via one of the greedy subsequences considered -- that, as explained below, covers a significant number of new elements when added to the solution. For any subset $A\subset [n]$, define \[
u(A)=(1-\alpha/k)\opt-|A|\,.
\]
Lemma~15 of Agrawal et al.~\cite{AgrawalSS19} establishes that if $C_{i-1}$ is the set of elements covered by the sets chosen during the processing of the first $i-1$ windows, and $T_{i-1}$ is the entire set of greedy subsequences constructed during the first $i-1$ windows, then\footnote{Note the  statement of  \cite[Lemma 15]{AgrawalSS19} has the RHS of Eq.~\eqref{eq:expectprogress} as $(1-\delta) e^{-\alpha/k} u(C_{i-1})$, rather than the expression in  Eq.~\eqref{eq:expectprogress}, where the $(1-\delta)$ appears in the exponent. Their proof actually only implies the weaker statement, in Eq.~\eqref{eq:expectprogress}; fortunately, this is still sufficient for their and our purposes.} 
\begin{equation}
\EE[u(S_1\cup \cdots \cup S_\alpha \cup C_{i-1}) \mid T_{i-1}] 
\leq
e^{-\alpha(1-\delta)/k} u(C_{i-1})\,, \label{eq:expectprogress}
\end{equation}
where $\delta=\sqrt{8/\beta}$, 
assuming $k\geq \alpha \beta$ and $\alpha\geq 4\beta^2 \log (\beta / 8)$. The crucial observation is that their proof of \cref{eq:expectprogress} holds true regardless of how $C_{i-1}$ is constructed from the greedy subsequences in $T_{i-1}$. In particular, Equation~\eqref{eq:expectprogress} holds true given our modification of the Agrawal et al.~algorithm, i.e., the possibility of adding additional sets from $\cR$ at the end of previous windows. At the end of the $i$th window, we potentially add additional sets $S_{\alpha+1}, \ldots, S_{\alpha'_i}$ where
\begin{align*}
&u(S_1\cup \cdots \cup S_{\alpha'_i} \cup C_{i-1}) \\
&\le u(S_1\cup \cdots \cup S_{\alpha} \cup C_{i-1}) -(\alpha'_i-\alpha)\opt/k\\
&\le (1-(\alpha'_i-\alpha)/k)  \cdot u(S_1\cup \cdots \cup S_{\alpha} \cup C_{i-1})\\
&\le e^{-(\alpha'_i-\alpha)/k} \cdot u(S_1\cup \cdots \cup S_{\alpha} \cup C_{i-1})
\end{align*}
Hence, in the $i$th window (except potentially in the window where we add the $k$th set to our solution) we add $\alpha'_i\geq \alpha$ sets $S_1, \ldots, S_{\alpha'_i}$ such that if $C_i=S_1\cup \cdots \cup S_{\alpha'}$
\begin{equation}
\EE[u(C_{i}) \mid T_{i-1}] 
\leq   
e^{-\alpha'_i (1-\delta) /k} \cdot u(C_{i-1}) \ .  \label{eq:expectprogress2}
\end{equation}
Since there are $k/\alpha$ windows, and we find a sequence of length $\alpha'_i>\alpha$ until we have a sequence of length $k$, such that Eq.~\eqref{eq:expectprogress2} holds in each window, for some $j\leq k/\alpha$, with $\delta<\epsilon/2<1$,
\begin{gather*}
  \EE[(1-\alpha/k)\opt - |C_{j}|] 
  \le (1-\alpha/k) \cdot \opt \cdot  
  \prod_{i\geq 1}e^{-\alpha'_i (1-\delta)/k} \\
  = (1-\alpha/k) e^{-1+\delta} \opt
  < (1-\alpha/k) (1/e+\epsilon/2)\opt 
\end{gather*}
where we used the fact $\sum_i \alpha'_i=k$ and $\exp(-1+x)\leq 1/e+x$ for $0<x<1$. Hence, $\EE[|C_j|]\leq 1-1/e-\epsilon$ assuming $\alpha/k \leq \epsilon/2$. Setting the constants to $\alpha = 4096\epsilon^{-4} \log (2/\epsilon)$ and $\beta=32 \epsilon^{-2}$
ensures all the necessary conditions are met assuming $k=\omega(1)$. With these values of $\alpha$ and $\beta$, we have $\binom{\alpha \beta}{\alpha} \leq \exp(O(\epsilon^{-4} \ln^2 \epsilon^{-1}))$. Hence, the space dependence on $\epsilon $ in Eq.~\ref{eq:randsize} becomes $O(k \log m\cdot \exp(O(\epsilon^{-4} \ln \epsilon^{-1})))$, since we may assume $\opt=O(k\epsilon^{-2} \log m)$ and the exponential dependencies on $\epsilon$ dominate the polynomial dependencies.

\section{Fast Algorithms}
\label{sec:fast}
\subsection{Fast Universe Sub-Sampling}

The universe sub-sampling approach described in  \Cref{sec:prelim}, requires the use of a $\gamma$-wise independent hash function where $\gamma=O(k \log m)$. This bound on the independence was actually an improvement by Jaud et al.~\cite{JaudWC23} over the   bound of $O(\epsilon^{-2} k \log m)$ originally given by McGregor and Vu~\cite{McGregorV19}. Jaud et al.~were motivated to improve the bound not for the sake of reducing the space of any streaming algorithms, but rather to reduce the update time from  $O(\epsilon^{-2} k \log m)$ to $O(k \log m)$. However, we note that the update time, at least in an amortized sense, can actually be reduced to polylogarithmic. Specifically, we can use the hash family given by degree-$\gamma$ polynomials over a suitably large finite field~\cite{WegmanC79}. While evaluating such a polynomial at a single point (which corresponds to hashing a single element) requires $\Omega(\gamma)$ time, evaluating the polynomial on $\gamma$ points can actually be done in $O(\gamma\log^2 \gamma \log \log \gamma)$ time~\cite[Chapter 10]{mda}. Hence, by buffering  sets of $O(k \log m)$ elements, we may ensure that the hash function can be applied with  only $\poly(\log k, \log \log m)$ amortized update time. We note that a similar approach was used by Kane et al.~ \cite{kane2011fast} in the context of frequency moment estimation.

\subsection{Fast $\mathbf{\ell_0}$ Sampling}
In the dynamic algorithm we needed to sample $r$ sets with replacement from a stream in the presence of insertions and deletions. This can be done via $\ell_0$ sampling as discussed in the preliminaries. However, a naive implementation of this process requires $\Omega(r)$ update time if $r$ different $\ell_0$ samplers need to be updated with every set insertion or deletion. To avoid this, a technique by McGregor et al.~\cite{McGregorTVV15} can be used.
Specifically, we use a $O(\log r)$-wise independent hash function to partition $[M]$ in $t:=r/\log r$ groups $P_1, P_2, \ldots P_{t}$. During the stream we compute \[\rho_i=|\{x \in P_i:\textstyle\sum_{j:x_j=x} \Delta_j\neq 0\}|\] 
i.e., the number of values in $P_i$ that are inserted a different number of times from the number of times they are deleted. Note that this is simple to compute $\rho_1, \ldots, \rho_t$ on the assumption that $\sum_{j:x_j=x} \Delta_j \in \{0,1\}$, i.e., every set is inserted either the same number of times it is deleted or exactly one extra time. We also compute $2\log r$ independent $\ell_0$ samplers for each $P_i$; note that each set insert or delete requires updating at most $2\log r$ independent $\ell_0$ samplers. Then at the end of the stream, to generate a sequence of samples with replacement we first sample $i$ with probability $\rho_i/\sum_j \rho_j$ and then use the next unused $\ell_0$ sampler from group $P_i$. Assuming $r\ll |\{x \in [M]:\sum_{j:x_j=x} \Delta_j\neq 0\}|$ (if this is not the case, we can just use sparse recovery), then with high  probability each $\rho_i/\sum_j \rho_j\leq 2(\log r)/r$ and the process will use at most $2\log r$ independent $\ell_0$ samplers from each group in expectation and less than $4\log r$ with probability at least $1-1/\poly(r)$.

\section{Conclusion}
\label{sec:conclusion}
We presented new algorithms for $1-1/e-\epsilon$ approximation of the maximum
coverage problem in the data stream model. These algorithms improve upon the
state-of-the-art results by a) reducing the space required in the dynamic
model from $(n+\epsilon^{-4}) \polylog(m,n)$ to $\epsilon^{-2} k
\polylog(m,n)$ when given $O(\epsilon^{-1} \log m)$ passes, b) reducing the
space required in single-pass random order model from $\epsilon^{-2} k^2
\polylog(m,n)$ to $O_\epsilon(k \polylog(m,n))$ although we emphasize that
the $O_\epsilon$ hides an exponential dependence on $\epsilon$, and c)
reducing the amortized update time to polylogarithmic in $m$ and $n$. We
conjecture that $\epsilon^{-2} k \polylog(m,n)$ space should be sufficient in
the random order model. In fact there is a monotone submodular optimization
algorithm in the random order model that, when applied to the maximum coverage
problem would maintain only  $O(\epsilon^{-1} k)$ sets \cite{LiuRVZ21}.  Given
that, via universe sub-sampling, we may assume that the sets have size
$O(\epsilon^{-2} k \log m)$; this yielded the Warneke et
al.~\cite{WarnekeCW23} result. However, it seems unlikely that the algorithm
of Liu et al.~\cite{LiuRVZ21} can be modified to ensure that the total number
of elements across all sets maintained is $O(\epsilon^{-2} k \log m)$ so a new
approach is likely to be necessary.

\bibliographystyle{alpha}
\bibliography{maxc}

\end{document}